\documentclass[aps,pra,reprint,superscriptaddress,nofootinbib,floatfix,amsmath,amssymb]{revtex4-2}

\usepackage{iftex}
\ifPDFTeX
  \usepackage[utf8]{inputenc}
  \usepackage[T1]{fontenc}
\fi
\usepackage{lmodern}
\usepackage{graphicx}
\usepackage{mathtools}
\usepackage{bbm}
\usepackage{microtype}
\usepackage{amsthm}
\usepackage[dvipsnames]{xcolor}
\usepackage{xurl}
\usepackage{hyperref}
\hypersetup{
  colorlinks=true,
  citecolor=blue,
  linkcolor=blue,
  urlcolor=blue,
  filecolor=black
}

\newtheorem{definition}{Definition}
\newtheorem{proposition}{Proposition}
\newtheorem{lemma}[proposition]{Lemma}

\newtheorem{theorem}[proposition]{Theorem}
\newtheorem{corollary}[proposition]{Corollary}

\newcounter{remark}

\newcounter{example}

\newcommand{\nc}{\newcommand}
\nc{\rnc}{\renewcommand}
\nc{\lbar}[1]{\overline{#1}}
\nc{\bra}[1]{\langle#1|}
\nc{\ket}[1]{|#1\rangle}
\nc{\ketbra}[2]{|#1\rangle\!\langle#2|}
\nc{\braket}[2]{\langle#1|#2\rangle}
\nc{\dketbra}[2]{\vert #1 \rangle \hspace{-.8mm} \rangle \hspace{-.4mm} \langle\hspace{-.8mm}\langle #2 \vert}
\nc{\dbra}[1]{\langle\hspace{-.8mm}\langle #1\vert}
\nc{\dket}[1]{\vert#1\rangle\hspace{-.8mm}\rangle}

\nc{\proj}[1]{| #1\rangle\!\langle #1 |}
\nc{\avg}[1]{\langle#1\rangle}
\nc{\smfrac}[2]{\mbox{$\frac{#1}{#2}$}}
\nc{\tr}{\operatorname{Tr}}
\nc{\ox}{\otimes}
\nc{\dg}{\dagger}
\nc{\dn}{\downarrow}
\nc{\cA}{{\cal A}}
\nc{\cB}{{\cal B}}
\nc{\cC}{{\cal C}}
\nc{\cD}{{\cal D}}
\nc{\cE}{{\cal E}}
\nc{\cF}{{\cal F}}
\nc{\cG}{{\cal G}}
\nc{\cH}{{\cal H}}
\nc{\cI}{{\cal I}}
\nc{\cJ}{{\cal J}}
\nc{\cK}{{\cal K}}
\nc{\cL}{{\cal L}}
\nc{\cM}{{\cal M}}
\nc{\cN}{{\cal N}}
\nc{\cO}{{\cal O}}
\nc{\cP}{{\cal P}}
\nc{\cQ}{{\cal Q}}
\nc{\cR}{{\cal R}}
\nc{\cS}{{\cal S}}
\nc{\cT}{{\cal T}}
\nc{\cU}{{\cal U}}
\nc{\cV}{{\cal V}}
\nc{\cX}{{\cal X}}
\nc{\cY}{{\cal Y}}
\nc{\cZ}{{\cal Z}}
\nc{\cW}{{\cal W}}

\nc{\bx}{\mathbf{x}}
\nc{\by}{\mathbf{y}}
\nc{\bz}{\mathbf{z}}

\nc{\csupp}{{\operatorname{csupp}}}
\nc{\qsupp}{{\operatorname{qsupp}}}
\nc{\var}{{\operatorname{var}}}
\nc{\rar}{\rightarrow}
\nc{\lrar}{\longrightarrow}
\nc{\polylog}{{\operatorname{polylog}}}
\nc{\wt}{{\operatorname{wt}}}
\nc{\supp}{{\operatorname{supp}}}

\nc{\argmin}{{\operatorname{argmin}}}
\nc{\bu}{\mathbf{u}}
\nc{\bv}{\mathbf{v}}
\nc{\eqt}[1]{\stackrel{\mathclap{\scriptsize \mbox{#1}}}{=}}
\nc{\leqt}[1]{\stackrel{\mathclap{\scriptsize \mbox{#1}}}{\leq}}
\nc{\geqt}[1]{\stackrel{\mathclap{\scriptsize \mbox{#1}}}{\geq}}

\nc{\RR}{{{\mathbb R}}}
\nc{\CC}{{{\mathbb C}}}
\nc{\FF}{{{\mathbb F}}}
\nc{\NN}{{{\mathbb N}}}
\nc{\ZZ}{{{\mathbb Z}}}
\nc{\PP}{{{\mathbb P}}}
\nc{\QQ}{{{\mathbb Q}}}
\nc{\EE}{{{\mathbb E}}}
\nc{\MM}{{{\mathbb M}}}
\nc{\id}{{\operatorname{id}}}

\nc{\SEP}{{\text{\rm SEP}}}
\nc{\NS}{{\text{\rm NS}}}
\nc{\LOCC}{{\text{\rm LOCC}}}
\nc{\PPT}{{\text{\rm PPT}}}
\nc{\EXT}{{\text{\rm EXT}}}
\nc{\Sym}{{\operatorname{Sym}}}
\nc{\clif}{\mathrm{Clif}}
\nc{\stab}{\mathrm{Stab}}

\allowdisplaybreaks

\newcommand{\revisionnote}[1]{}
\begin{document}

\title{Amortized Stabilizer R\'enyi Entropy of Quantum Dynamics}

\author{Chengkai Zhu}
\affiliation{Thrust of Artificial Intelligence, Information Hub, The Hong Kong University of Science and Technology (Guangzhou), Guangzhou 511453, China}

\author{Yu-Ao Chen}
\affiliation{Thrust of Artificial Intelligence, Information Hub, The Hong Kong University of Science and Technology (Guangzhou), Guangzhou 511453, China}

\author{Zhiping Liu}
\affiliation{Thrust of Artificial Intelligence, Information Hub, The Hong Kong University of Science and Technology (Guangzhou), Guangzhou 511453, China}

\author{Zhan Yu}
\affiliation{Centre for Quantum Technologies, National University of Singapore, 117543, Singapore}

\author{Zanqiu Shen}
\affiliation{Thrust of Artificial Intelligence, Information Hub, The Hong Kong University of Science and Technology (Guangzhou), Guangzhou 511453, China}

\author{Xin Wang}
\email{felixxinwang@hkust-gz.edu.cn}
\affiliation{Thrust of Artificial Intelligence, Information Hub, The Hong Kong University of Science and Technology (Guangzhou), Guangzhou 511453, China}

\begin{abstract}
Magic states are essential resources for fault-tolerant quantum computation. We introduce the amortized stabilizer R\'enyi entropy (SRE), which quantifies the largest SRE increase generated by a unitary, including arbitrary pure inputs and reference systems. Initial nonstabilizerness can enhance this increase, whereas analogous amortized logarithmic robustness and extent collapse to stabilizer-input optimizations. We establish structural properties and exact amortized 2-SRE values for the $T$ and $CCZ$ gates. We then derive explicit lower bounds on approximate $T$-state cost, with diamond-distance error, allowing stabilizer ancillary states, Pauli measurements, and classical feedforward. A Choi-state witness gives computable bounds for quantum Fourier transforms, controlled rotations, fractional controlled-SWAP gates, and finite Heisenberg-chain evolutions. We distinguish the unrestricted optimization from its Choi witness and specify the access and precision requirements for estimating Pauli moments. An odd-moment variant yields a confidence-qualified cost witness from ordinary copies.
\end{abstract}


\maketitle

\section{Introduction}
The pursuit of quantum technologies is driven by the promise of achieving a significant quantum speed-up over classical systems in various domains, such as computing~\cite{Shor_1997,grover1996fast,RevModPhysAndrew} and quantum simulation~\cite{Lloyd1996,Childs_2018}. To harness the full potential of quantum technologies and realize this advantage, the \textit{nonstabilizerness} resources become essential since all \textit{stabilizer operations} can be efficiently simulated on a classical computer, known as the Gottesman-Knill theorem~\cite{Aaronson_2004, gottesman1997stabilizer}.
These nonstabilizerness resources are referred to as \textit{magic states}, which cannot be prepared using stabilizer operations~\cite{Veitch2014}. Clifford operations admit economical fault-tolerant implementations in many quantum error-correcting codes~\cite{shor1997faulttolerant,Campbell_2017,knill2005quantum}. Suitable magic states enable non-Clifford operations through state injection, completing a universal gate set~\cite{Zhou_2000,Gottesman_1999,Bravyi_2005}.

Therefore, understanding the nonstabilizerness of quantum states becomes crucial and motivates extensive research on the quantum resource theory of magic states~\cite{Veitch2014,Howard2017,Wang2018,takagi2018convex,ahmadi2018quantification,Leone2024}.
Among different approaches for quantifying nonstabilizerness resources~\cite{regula2017convex,Howard2017,Bravyi_2019}, the \textit{Stabilizer R{\'e}nyi Entropy} (SRE)~\cite{Leone2022} has recently been proposed, emerging as a powerful tool for studying the nonstabilizerness of pure states. Their associated Pauli moments admit efficient additive estimation under specified measurement-access assumptions~\cite{Haug2024}, while taking logarithms can require substantially greater precision. Tensor-network methods also allow their calculation for matrix product states (MPSs) with controlled bond dimension~\cite{Haug2023b,Lami2023,Tarabunga2024}, thus promoting the study of nonstabilizerness in information scrambling~\cite{Ahmadi2024}, random quantum circuits~\cite{Niroula_2024,Turkeshi_2025}, property estimation~\cite{Leone_2023,Lorenzo2024a}, many-body systems~\cite{Oliviero2022a,Davide2023,Liu2024,Frau2024,Chen2024,Lami2025}, as well as the intrinsic relationship between entanglement and nonstabilizerness~\cite{Gu2024,Gu2024a}.

Quantum dynamics can generate nonstabilizerness and thereby affect the cost of stabilizer-based classical simulation. Characterizing this dynamical resource requires accounting for its action on different inputs and reference systems. To tackle this challenge, several approaches have been proposed to give dynamical magic measures, including the mana and generalized thauma of a quantum channel~\cite{WWS19}, the channel robustness of magic~\cite{Seddon2019}, the entropic and geometric measures~\cite{Saxena2022}. These measures provide valuable insights into the nonstabilizerness of quantum dynamics, e.g., the classical simulability and the distance from the stabilizer operations.

Can we quantify the nonstabilizerness of quantum dynamics in a more natural and intuitive way? An alternative approach is to determine how much nonstabilizerness a quantum dynamics can generate at most, which captures the nonstabilizerness generation capability of the dynamics. This concept is reminiscent of the so-called \textit{amortized resources} in other quantum resource theories, such as the amortized entanglement (or entangling capacity)~\cite{Bennett_2003,Campbell2010,Berta2017a,Kaur_2017,Wang_2023}, and the amortized coherence~\cite{Ben_Dana_2017,D_az_2018,Garc2015}. Notably, it has been shown that the maximum entanglement increase of a unitary can be enhanced when the initial two-qubit state possesses prior entanglement~\cite{Leifer2003}. Building upon the success of SREs in quantifying the nonstabilizerness of pure states, it is natural and intriguing to ask what the maximum possible increase in SREs of a unitary operation would be; would initial nonstabilizerness in the input state boost such nonstabilizerness generation of the unitary?

In this work, we answer this question affirmatively by demonstrating that the nonstabilizerness generation can indeed be boosted if the input states possess prior nonstabilizerness, compared to stabilizer states as inputs. This finding motivates our investigation into the amortized magic of multi-qubit unitary operations, which serves as a magic measure for unitary operations, with a particular focus on the \textit{amortized $\alpha$-stabilizer R\'enyi entropy} of an $n$-qubit unitary operation. Intriguingly, we also show that the presence of prior nonstabilizerness in input states does not provide an advantage when considering alternative measures such as the amortized magic robustness~\cite{Howard2017} or the amortized stabilizer extent~\cite{Bravyi_2019}.

As an application, we first establish the exact Clifford+$T$ lower bound and then derive an error-dependent extension for deterministic adaptive protocols. This distinction is necessary because generic rotations and Hamiltonian evolutions need not be exactly synthesizable. We give both a uniform bound $M_{2}^{\cA}(U)\leq2n$ and a continuity estimate that does not depend on the ancillary dimension. The latter turns Choi moments into lower bounds at explicitly specified synthesis error. We also explain how odd-order moments provide experimentally accessible, confidence-qualified witnesses, and separate these from the generally difficult unrestricted amortized optimization.

\section{Amortized magic}\label{sec:amortized}

We consider normalized pure states on $\cH_{n}=(\CC^{2})^{\ox n}$ and use base-two logarithms. Let $\cP_{n}=\{I,X,Y,Z\}^{\ox n}$ denote the Hermitian representatives of the Pauli group modulo phases. We reserve $\stab_{n}$ for normalized pure stabilizer states and write
\[
\mathrm{STAB}_{n}=\operatorname{conv}\{\ketbra{s}{s}:\ket{s}\in\stab_{n}\}
\]
for their convex hull. Clifford unitaries normalize the Pauli group; global phases are immaterial throughout. We use $\NN_{0}=\{0,1,2,\ldots\}$, so the ancillary system may be absent.
For $0\leq\alpha<\infty$, the stabilizer R\'enyi entropy is~\cite{Leone2022}
\begin{equation}\label{eq:SREdef}
\begin{aligned}
 M_{\alpha}(\ket{\psi})&=\frac{1}{1-\alpha}\log_{2}
 \sum_{P\in\cP_{n}}\Xi_{P}(\ket{\psi})^{\alpha}-n,\\
 \Xi_{P}(\ket{\psi})&=\frac{|\bra{\psi} P\ket{\psi}|^{2}}{2^{n}}.
\end{aligned}
\end{equation}
Here $\sum_{P}\Xi_{P}=1$; $\alpha=1$ denotes the continuous limit and $\alpha=0$ counts nonzero probabilities. These quantities are faithful, Clifford invariant, and additive on tensor products. For integer $\alpha\geq2$, they are monotones for deterministic pure-state stabilizer protocols~\cite{Leone2024,Haug2023a}. The unitary properties below follow from faithfulness, Clifford invariance, and ancillary stability, and hold for all finite $\alpha\geq0$.

The initial nonstabilizerness in the input state can boost SRE generation. Consider $\sqrt{T}=\operatorname{diag}(1,e^{i\pi/8})$. Exhausting the six single-qubit stabilizer inputs gives
\[
 \max_{\ket{\phi}\in\stab_{1}}M_{2}(\sqrt{T}\ket{\phi})
 =M_{2}(\sqrt{T}\ket{+}) =3-\log_{2} 7.
\]
For $\ket{\psi}=(\ket{0}+e^{i\pi/10}\ket{1})/\sqrt{2}$, direct calculation gives
\begin{align*}
 M_{2}(\ket{\psi})&=3-\log_{2}[7+\cos(2\pi/5)],\\
 M_{2}(\sqrt{T}\ket{\psi})&=3-\log_{2}[7+\cos(9\pi/10)].
\end{align*}
Their difference exceeds $3-\log_{2} 7$. This resource dependence is reminiscent of the enhancement of entanglement generation by initially entangled states~\cite{Leifer2003,Campbell2010}.

We assume that a pure-state magic measure $\MM$ satisfies, in every finite dimension: (i) $\MM(\ket{\psi})\geq0$, with $\MM(\ket{\psi})=0$ \emph{if and only if} $\ket{\psi}\in\stab_{n}$; (ii) Clifford invariance; and (iii) invariance under adjoining a pure stabilizer state, $\MM(\ket{\psi}\ox\ket{s})=\MM(\ket{\psi})$. SRE, log robustness, and log extent satisfy these assumptions.

\begin{definition}[Amortized magic]\label{def:amortized}
For an $n$-qubit unitary $U$, define
\begin{equation}\label{eq:amortized}
 \MM^{\cA}(U)=\sup_{m\in\NN_{0}}\sup_{\ket{\psi}\in\cH_{n+m}}\big\{\MM((U\ox I_{2^{m}})\ket{\psi})-\MM(\ket{\psi})\big\}.
\end{equation}
The strict (stabilizer-input) version is
\begin{equation}
 \widetilde{\MM}^{\cA}(U)=\sup_{m\in\NN_{0}}
 \max_{\ket{s}\in\stab_{n+m}}\MM((U\ox I_{2^{m}})\ket{s}).
\end{equation}
All pure-state optimizations are over normalized vectors.
\end{definition}

Let $\ket{\Phi_{n}}=2^{-n/2}\sum_{j=0}^{2^{n}-1}\ket{j}\ket{j}$ and
$\ket{\Phi_{U}}=(U\ox I_{2^{n}})\ket{\Phi_{n}}$ be the normalized Choi state. Since $\ket{\Phi_{n}}$ is a stabilizer state,
\begin{equation}\label{eq:choi_chain}
 \MM^{\cA}(U)\geq\widetilde{\MM}^{\cA}(U)\geq\MM(\ket{\Phi_{U}}).
\end{equation}
The last quantity is a particular input witness, whereas the first optimizes both the input and the ancillary dimension. Amortized resource constructions for channels have also been studied in~\cite{WWS19,Gilad2019,Liu2019}.

\begin{lemma}[Faithfulness]\label{lem:faithful}
$\MM^{\cA}(U)=0$ if and only if $U$ is Clifford, up to global phase.
\end{lemma}
\begin{proof}
If $U$ is Clifford, Clifford invariance makes every magic difference zero. Conversely, $\MM^{\cA}(U)=0$ and Eq.~\eqref{eq:choi_chain} imply that $\ket{\Phi_{U}}$ is a stabilizer state. A maximally entangled bipartite stabilizer state is related to $n$ Bell pairs by local Clifford unitaries~\cite{fattal2004entanglement}. Thus $\ket{\Phi_{U}}=(C_{A}\ox C_{B})\ket{\Phi_{n}}$ up to phase. The Choi identity gives $U=C_{A} C_{B}^{T}$ up to phase, which is Clifford.
\end{proof}

\begin{lemma}[Subadditivity under composition]\label{lem:subadd_comp}
For unitaries on the same system,
$\MM^{\cA}(UV)\leq\MM^{\cA}(U)+\MM^{\cA}(V)$.
\end{lemma}
\begin{proof}
For every $m$ and $\ket{\psi}$, set $\ket{\phi}=(V\ox I)\ket{\psi}$ and split the increase as
\begin{align*}
 &\MM((UV\ox I)\ket{\psi})-\MM(\ket{\psi})\\
 &=\MM((U\ox I)\ket{\phi})-\MM(\ket{\phi})
   +\MM((V\ox I)\ket{\psi})-\MM(\ket{\psi})\\
 &\leq\MM^{\cA}(U)+\MM^{\cA}(V).
\end{align*}
Taking the supremum proves the claim.
\end{proof}

\begin{lemma}[Clifford invariance]\label{lem:inv_clif}
\leavevmode\par\nobreak\noindent
For Clifford unitaries $C_{1},C_{2}$, $\MM^{\cA}(C_{1}UC_{2})=\MM^{\cA}(U)$.
\end{lemma}
\begin{proof}
Apply Lemma~\ref{lem:subadd_comp} and faithfulness in both directions, using $U=C_{1}^{\dagger}(C_{1}UC_{2})C_{2}^{\dagger}$.
\end{proof}

\begin{lemma}[Subadditivity under tensor product]\label{lem:tensor}
$\MM^{\cA}(U\ox V)\leq\MM^{\cA}(U)+\MM^{\cA}(V)$.
\end{lemma}
\begin{proof}
First, $\MM^{\cA}(U\ox I_{2^{k}})=\MM^{\cA}(U)$: one direction follows by absorbing $k$ qubits into the reference, and the other by adjoining $k$ stabilizer qubits to any candidate input. Their addition preserves the magic difference by assumption (iii). Clifford SWAP gates reorder registers if necessary. Now apply composition subadditivity to $(U\ox I)(I\ox V)$.
\end{proof}
Consequently, amortized magic is invariant under Clifford pre- and postprocessing and adjoining Clifford unitaries. The extension needed for adaptive synthesis is proved separately in Sec.~\ref{sec:approx}.

\section{Amortized stabilizer R\'enyi entropy}\label{sec:sre}

\begin{definition}[Amortized stabilizer R\'enyi entropy]
\leavevmode\par\nobreak\noindent
The amortized $\alpha$-SRE $M_{\alpha}^{\cA}(U)$ and its strict version $\widetilde{M}_{\alpha}^{\cA}(U)$ are Definition~\ref{def:amortized} with $\MM=M_{\alpha}$, for finite $0\leq\alpha<\infty$.
\end{definition}
These quantities obey the unitary resource properties of Sec.~\ref{sec:amortized}.

\begin{proposition}\label{prop:SMA_meqn}
For an $n$-qubit unitary, the strict optimization requires at most $n$ reference qubits:
\begin{equation}
 \widetilde{M}_{\alpha}^{\cA}(U)=\max_{\ket{s}\in\stab_{2n}}
 M_{\alpha}((U\ox I_{2^{n}})\ket{s}).
\end{equation}
\end{proposition}
\begin{proof}
For $m<n$, pad a candidate stabilizer input by $n-m$ stabilizer qubits. For $m>n$, the bipartite stabilizer normal form~\cite{fattal2004entanglement} implies
\[
 \ket{s}_{AB}=(I_{A}\ox C_{B})(\ket{r}_{AB'}\ox\ket{\tau}_{B''}),
\]
where $A,B'$ each contain $n$ qubits, $B''$ contains $m-n$, $C_{B}$ is Clifford, and both $\ket{r}$ and $\ket{\tau}$ are stabilizer states. Since $U$ commutes with $I_{A}\ox C_{B}$, Clifford invariance and tensor additivity remove $C_{B}$ and $\tau$ from the output SRE. This proves the reduction. The argument is also used for channel magic capacity in~\cite{Seddon2019}.
\end{proof}

\begin{definition}[Log robustness of magic]\label{def:ROM_states}
For a density operator $\rho$, define
\begin{equation}
 \cR(\rho)=\min_{\{q_{j},s_{j}\}}\sum_{j}|q_{j}|,
 \qquad \rho=\sum_{j}q_{j}\ketbra{s_{j}}{s_{j}},
\end{equation}
where $q_{j}\in\RR$ and $\ket{s_{j}}\in\stab_{n}$. The log robustness is $\log_{2}\cR(\rho)$~\cite{Howard2017,Seddon2019}. For pure vectors we use $\cR(\ket{\psi})$ as shorthand for $\cR(\ketbra{\psi}{\psi})$.
\end{definition}

\begin{definition}[Stabilizer extent]\label{def:stab_extent}
For a normalized pure state,
\begin{equation}
 \xi(\ket{\psi})=\min_{\{c_{j},s_{j}\}}\left(\sum_{j}|c_{j}|\right)^{2},
 \qquad\ket{\psi}=\sum_{j}c_{j}\ket{s_{j}},
\end{equation}
where $c_{j}\in\CC$ and $\ket{s_{j}}\in\stab_{n}$. The logarithmic measure is $\log_{2}\xi$~\cite{Bravyi_2019}.
\end{definition}
Applying Definition~\ref{def:amortized} to these two pure-state measures gives $\cR_{\log}^{\cA}$, $\xi_{\log}^{\cA}$, and their strict versions. Both measures are Clifford invariant and unchanged by adjoining a pure stabilizer state. The reverse of the ancillary-padding inequality for extent follows by contracting any stabilizer decomposition with that ancillary stabilizer bra; each contraction is a subnormalized stabilizer vector. Thus the reference-compression argument above also applies to their strict versions.

\begin{proposition}\label{pro:str_amo_log_RoM}
$\cR_{\log}^{\cA}(U)=\widetilde{\cR}_{\log}^{\cA}(U)$.
\end{proposition}
\begin{proof}
For fixed $m$, let $C_{m}=\max_{s\in\stab_{n+m}}\cR((U\ox I)\ket{s})$. An optimal real stabilizer decomposition of a pure input projector $\psi$ gives, by the triangle inequality for robustness,
\[
 \cR((U\ox I)\psi(U^{\dagger}\ox I))
 \leq\sum_{j}|q_{j}|\cR((U\ox I)\ket{s_{j}})
 \leq C_{m}\cR(\psi).
\]
Taking logarithms and suprema yields the upper bound by $\sup_{m}\log_{2} C_{m}$. Stabilizer inputs attain the reverse inequality because their robustness is one.
\end{proof}

\begin{proposition}\label{pro:str_amo_stab_extent}
$\xi_{\log}^{\cA}(U)=\widetilde{\xi}_{\log}^{\cA}(U)$.
\end{proposition}
\begin{proof}
Fix $m$ and set $K_{m}=\max_{s\in\stab_{n+m}}\xi((U\ox I)\ket{s})$. Choose optimal decompositions
\[
 \ket{\psi}=\sum_{j}c_{j}\ket{s_{j}},\qquad
 (U\ox I)\ket{s_{j}}=\sum_{i}\beta_{ij}\ket{s_{ij}}.
\]
The optimal coefficient sums satisfy
\[
 \sum_{j}|c_{j}|=\sqrt{\xi(\psi)},\qquad
 \sum_{i}|\beta_{ij}|=\sqrt{\xi((U\ox I)s_{j})}.
\]
The composed expansion is a feasible stabilizer decomposition, so
\begin{align}\label{eq:amo_stab_extent}
 \sqrt{\xi((U\ox I)\psi)}
 &\leq\sum_{i,j}|c_{j}\beta_{ij}|\notag\\
 &\leq\sqrt{K_{m}}\sum_{j}|c_{j}|
 =\sqrt{K_{m}\xi(\psi)}.
\end{align}
Squaring gives $\xi((U\ox I)\psi)\leq K_{m}\xi(\psi)$. Logarithms and suprema over $\psi,m$ prove one direction; stabilizer inputs give equality in the comparison of the two capacities.
\end{proof}
The collapse of these two amortized measures follows from their decomposition norms. The equatorial example below shows that such a collapse fails for the 2-SRE.

For $R_{z}(\theta)=e^{-i\theta Z/2}$, an explicit equatorial input already gives a nontrivial amortized witness. Write $\ket{+_{\varphi}}=(\ket{0}+e^{i\varphi}\ket{1})/\sqrt{2}$ and $P_{2}(\psi)=2^{-M_{2}(\psi)}$. Then $P_{2}(+_{\varphi})=(7+\cos4\varphi)/8$, and maximizing the ratio $P_{2}(+_{\varphi})/P_{2}(+_{\varphi+\theta})$ gives
\begin{equation}\label{eq:Rz_witness}
 M_{2}^{\cA}(R_{z}(\theta))\geq
 \log_{2}\frac{a+\sqrt{a^{2}-48^{2}}}{48},
 \qquad a=49-\cos4\theta.
\end{equation}
Indeed an extremal ratio $r$ obeys $49(1-r)^{2}=1-2r\cos4\theta+r^{2}$; the larger root yields Eq.~\eqref{eq:Rz_witness}. Optimizing over arbitrary inputs and references remains an open problem for general $\theta$.

For $0\leq\theta\leq\pi/2$, the strict 2-SRE is $-\log_{2}[(7+\cos4\theta)/8]$. To see that a one-qubit reference cannot improve it, use the bipartite normal form of two-qubit stabilizer states: product inputs reduce to the six single-qubit stabilizer states, and entangled inputs are local-Clifford equivalents of a Bell pair. A Clifford on the acted-on half of a Bell pair can be transposed onto the reference. Moreover, the Choi state of $R_{z}(\theta)$ is Clifford equivalent to $\ket{+_{\theta}}\ox\ket{0}$.
The same reasoning and the single-qubit robustness and extent formulas~\cite{Howard2017,Bravyi_2019} give
\begin{align}
 \cR_{\log}^{\cA}(R_{z}(\theta))&=\log_{2}(\cos\theta+\sin\theta),\label{eq:Rz_ROM}\\
 \xi_{\log}^{\cA}(R_{z}(\theta))&=\log_{2}[1+(\sqrt{2}-1)(\cos\theta+\sin\theta-1)].\label{eq:Rz_extent}
\end{align}
For completeness, the extent formula follows by decomposing $\ket{+_{\theta}}$ into $\ket{+}$ and $\ket{+_{\pi/2}}$. Their coefficient magnitudes are $\sqrt{2}\sin(\pi/4-\theta/2)$ and $\sqrt{2}\sin(\theta/2)$. A matching dual vector is
$\ket{w}=(\ket{0}+e^{i\pi/4}\ket{1})/(\sqrt{2}\cos(\pi/8))$, whose overlap magnitude with each of the six single-qubit stabilizer states is at most one. This certifies the squared sum of the coefficient magnitudes. Figure~\ref{fig:MARz} uses these analytic expressions.

\begin{figure}[!tbp]
\centering
\includegraphics[width=\linewidth]{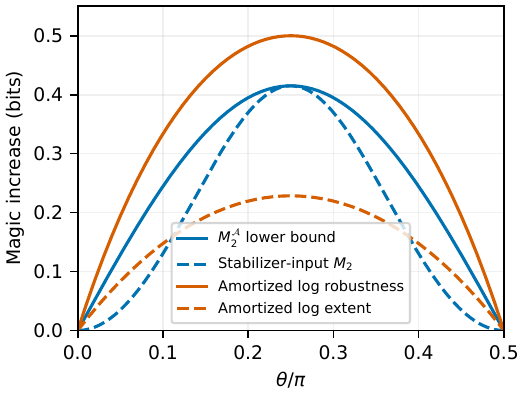}
\caption{Magic generation for $R_{z}(\theta)=e^{-i\theta Z/2}$. Solid blue: the achievable amortized 2-SRE witness in Eq.~\eqref{eq:Rz_witness}. Dashed blue: the strict amortized 2-SRE. Orange: amortized log robustness (solid) and log extent (dashed). These latter two quantities equal their strict versions. All curves follow analytic formulas.}
\label{fig:MARz}
\end{figure}

\begin{theorem}\label{thm:amor_magic_T}
For $T=\operatorname{diag}(1,e^{i\pi/4})$, $M_{2}^{\cA}(T)=2-\log_{2} 3$.
\end{theorem}
\begin{proof}
The input $\ket{+}$ gives $P_{2}(T\ket{+})=3/4$ and hence the required lower bound. For the upper bound, let
$p_{jQ}=\bra{\psi}(\sigma_{j}\ox Q)\ket{\psi}$ for any pure state on $1+m$ qubits. Set $R_{2}(\psi)=\sum_{j,Q}p_{jQ}^{4}$, so $M_{2}(\psi)=1+m-\log_{2}R_{2}(\psi)$. Since
$T^{\dagger} XT=(X-Y)/\sqrt{2}$ and $T^{\dagger} YT=(X+Y)/\sqrt{2}$,
\begin{align*}
&4R_{2}((T\ox I)\psi)-3R_{2}(\psi)\\
&=\sum_{Q}(p_{0Q}^{4}+p_{3Q}^{4}-p_{1Q}^{4}-p_{2Q}^{4}+12p_{1Q}^{2}p_{2Q}^{2})\geq0.
\end{align*}
The last step follows from Lemma~\ref{lem:IZgeXY} below. Thus the output fourth moment is at least $3/4$ of the input fourth moment, independently of $m$. Taking logarithms proves the upper bound.
\end{proof}

\begin{lemma}\label{lem:P4>=0}
The four-copy operators $Q_{+}=I^{\ox4}+X^{\ox4}+Y^{\ox4}+Z^{\ox4}$ and $Q_{-}=I^{\ox4}-X^{\ox4}-Y^{\ox4}+Z^{\ox4}$ are positive semidefinite.
\end{lemma}
\begin{proof}
$X^{\ox4}$ and $Z^{\ox4}$ are commuting Hermitian involutions, and $Y^{\ox4}=X^{\ox4}Z^{\ox4}$. Hence
$Q_{\pm}=(I^{\ox4}\pm X^{\ox4})(I^{\ox4}+Z^{\ox4})\succeq0$.
\end{proof}
\begin{lemma}\label{lem:IZgeXY}
For every $\ket{\psi}\in\cH_{1+m}$,
$\sum_{Q}(p_{0Q}^{4}+p_{3Q}^{4}-p_{1Q}^{4}-p_{2Q}^{4})\geq0$.
\end{lemma}
\begin{proof}
After permuting tensor factors, this sum equals the expectation of $Q_{-}\ox Q_{+}^{\ox m}$ in $\ket{\psi}^{\ox4}$, and is nonnegative by Lemma~\ref{lem:P4>=0}.
\end{proof}
We also prove $M_{2}^{\cA}(CCZ)=5-\log_{2} 11$ in Appendix~\ref{app:amor_magic_CCZ}. For $T$ and $CCZ$ the optimal inputs can be chosen as $\ket{+}$ and $\ket{+}^{\ox3}$, respectively. These stabilizer inputs attain the maxima over arbitrary inputs and references. Equation~\eqref{eq:Rz_witness} demonstrates the role of initial magic for other rotations.

\section{Exact and approximate T-count}\label{sec:tcount}

\subsection{Exact synthesis}
Let $t(U)$ be the minimum number of $T$ or $T^{\dagger}$ gates in an ancilla-free Clifford+$T$ circuit implementing $U$ up to global phase; set $t(U)=+\infty$ when no such circuit exists. Counting $T^{\dagger}$ at unit cost is consistent with $T^{\dagger}=S^{\dagger} T$. This cost is relevant to fault-tolerant implementation and classical simulation~\cite{Bravyi_2016,Bravyi_2019,Seddon2021}. Generic continuous parameters need not admit exact synthesis, so an exact lower bound can be vacuous. For fixed $n$, finite Clifford+$T$ circuits form a countable set up to phase, whereas unitary gates form a continuous family; universality guarantees arbitrarily accurate approximation, rather than finite exact synthesis of every target. We address approximate and adaptive synthesis below.
\begin{theorem}[Exact circuit bound]\label{thm:Tcount_lowbound}
For a magic measure obeying the assumptions of Definition~\ref{def:amortized}, with $0<\MM^{\cA}(T)<\infty$,
\begin{equation}
 t(U)\geq\frac{\MM^{\cA}(U)}{\MM^{\cA}(T)}.
\end{equation}
\end{theorem}
\begin{proof}
For any $n$-qubit unitary $U$ with $k=t(U)<\infty$, an optimal exact circuit can be written, up to global phase, as
\[
 U=C_k V C_{k-1}\cdots C_1 V C_0,
 \qquad V=T\ox I_{2^{n-1}},
\]
where $C_0,\ldots,C_k$ are Clifford gates; for $k=0$, this means $U=C_0$. The Clifford factor in $T^{\dagger}=S^{\dagger}T$ and Clifford SWAPs are absorbed into the $C_j$. By subadditivity under composition and faithfulness (Lemmas~\ref{lem:subadd_comp} and~\ref{lem:faithful}),
\[
 \MM^{\cA}(U)\leq k\,\MM^{\cA}(T\ox I_{2^{n-1}}).
\]
The ancillary-invariance argument in the proof of Lemma~\ref{lem:tensor} gives $\MM^{\cA}(T\ox I_{2^{n-1}})=\MM^{\cA}(T)$. Dividing by the positive denominator proves the bound. If no finite exact circuit exists, $t(U)=+\infty$ and the inequality is automatic.
\end{proof}
\begin{corollary}\label{cor:M2U_Tlowerbound}
Writing $c_{T}=\log_{2}(4/3)$, we have
\begin{equation}\label{eq:exact_choi_bound}
 t(U)\geq\left\lceil\frac{M_{2}^{\cA}(U)}{c_{T}}\right\rceil
 \geq\left\lceil\frac{M_{2}(\ket{\Phi_{U}})}{c_{T}}\right\rceil.
\end{equation}
\end{corollary}
Composition subadditivity gives the cost bound, and fixing the input to the maximally entangled stabilizer state gives the Choi witness.

The $O(n^{2})$ synthesis bound of Amy and Mosca~\cite{Amy_2019} concerns the CNOT+$T$ phase-polynomial class. A direct bound on amortized SRE follows from Pauli norm inequalities.
\begin{proposition}[Uniform dimension bound]\label{prop:dimension_bound}
\leavevmode\par\nobreak\noindent
For every $n$-qubit unitary, $0\leq M_{2}^{\cA}(U)\leq2n$.
\end{proposition}
\begin{proof}
For a fixed reference Pauli $Q$, collect the $4^{n}$ coefficients
$v_{Q}(P)=\bra{\psi}(P\ox Q)\ket{\psi}$. Conjugation by $U$ acts on this real vector through an orthogonal Pauli transfer matrix $O$. Norm inequalities give
\[
 \|Ov_{Q}\|_{4}^{4}\geq4^{-n}\|Ov_{Q}\|_{2}^{4}
 =4^{-n}\|v_{Q}\|_{2}^{4}\geq4^{-n}\|v_{Q}\|_{4}^{4}.
\]
Sum over $Q$ and take the negative logarithm of the fourth-moment ratio. The estimate is uniform in the reference dimension and yields $M_{2}^{\cA}(U)\leq\log_{2} 4^{n}=2n$.
\end{proof}

\subsection{Approximation, ancillary states, and measurements}\label{sec:approx}
For $0\leq\epsilon\leq1$, define $t^{\rm stab}_{\epsilon}(U)$ as the smallest \emph{worst-case} number of $T$-state resources used by a deterministic protocol whose implemented channel $\cE$ satisfies
\begin{equation}\label{eq:diamond_convention}
 \tfrac{1}{2}\|\cE-\cU\|_{\diamond}\leq\epsilon,
 \qquad \cU(\rho)=U\rho U^{\dagger}.
\end{equation}
Set $t^{\rm stab}_{\epsilon}(U)=+\infty$ if no such finite protocol exists.
We consider finite, deterministic protocols with $n$-qubit input and output registers. Free operations comprise stabilizer ancillary preparations, Clifford unitaries, Pauli measurements, classical randomness and feedforward, and discarding. The protocol starts with any finite number of stabilizer work qubits and at most $t$ copies of $\ket{T}=T\ket{+}$ as its supplied nonstabilizer resources. One such state implements one $T$ gate by stabilizer state injection~\cite{Bravyi_2005,Zhou_2000}. The budget applies to every branch, and the final channel averages over all measurement outcomes.

For an $N$-qubit pure state define
\begin{equation}\label{eq:Palpha}
 P_{\alpha}(\psi)=2^{-N}\sum_{P\in\cP_{N}}|\bra{\psi} P\ket{\psi}|^{2\alpha}
 =2^{(1-\alpha)M_{\alpha}(\psi)}.
\end{equation}
The proof uses the mixed-state extension of Leone and Bittel~\cite{Leone2024},
\begin{equation}\label{eq:roof}
\begin{aligned}
 \widehat{P}_{2}(\rho)&=\sup_{\rho=\sum_{j}p_{j}\phi_{j}}\sum_{j}p_{j}P_{2}(\phi_{j}),\\
 \widehat{M}_{2}(\rho)&=-\log_{2}\widehat{P}_{2}(\rho),
\end{aligned}
\end{equation}
where $\phi_{j}$ are pure-state projectors. Strong monotonicity of $1-P_{2}$~\cite[Theorems 2 and 3]{Leone2024} gives deterministic monotonicity of this extension as follows. Given any pure input $\psi_{i}$, refine all classical choices and replace each discarded qubit by a computational-basis measurement with an unobserved outcome. The protocol then yields a pure-state ensemble $\{p_{j|i},\phi_{ij}\}$ with
\[
 P_{2}(\psi_{i})\leq\sum_{j}p_{j|i}P_{2}(\phi_{ij}).
\]
For any decomposition $\rho=\sum_{i}q_{i}\psi_{i}$, these outputs form a decomposition of $\Lambda(\rho)$. Therefore
\[
 \widehat{P}_{2}(\Lambda(\rho))\geq\sum_{i}q_{i}P_{2}(\psi_{i}).
\]
Taking the input supremum proves $\widehat{P}_{2}(\Lambda(\rho))\geq\widehat{P}_{2}(\rho)$ and hence deterministic monotonicity of $\widehat{M}_{2}$. This argument includes forgetting measurement records and discarding unused resource states.

\begin{lemma}[Continuity near a pure target]\label{lem:continuity}
For pure-state projectors $\psi,\phi$ on the same number of qubits,
\begin{equation}\label{eq:pure_continuity}
 |P_{2}(\psi)-P_{2}(\phi)|\leq4D(\psi,\phi),
 \quad D(\rho,\sigma)=\tfrac{1}{2}\|\rho-\sigma\|_{1}.
\end{equation}
For any density operator $\rho$ with $D(\rho,\phi)\leq\epsilon$,
\begin{equation}\label{eq:roof_continuity}
 \widehat{P}_{2}(\rho)\leq\min\{1,P_{2}(\phi)+16\epsilon\}.
\end{equation}
\end{lemma}
\begin{proof}
Let $d=2^{N}$ and, for density operators $\omega$, define
\[
 f(\omega)=d^{-1}\sum_{P}[\tr(P\omega)]^{4},\qquad
 G_{\omega}=d^{-1}\sum_{P}[\tr(P\omega)]^{3}P.
\]
The polynomial $f$ extends the pure-state expression $P_{2}$ to density operators. We first show $0\preceq G_{\omega}\preceq I$. Indeed, $\tr G_{\omega}=1$, and for any vector $v$,
\[
 \bra{v}G_{\omega}\ket{v}=d^{-1}\tr\!\left[
 (\omega^{\ox3}\ox\ketbra{v}{v})\sum_{P}P^{\ox4}\right]\geq0.
\]
After regrouping replicas, $\sum_{P}P^{\ox4}=Q_{+}^{\ox N}\succeq0$ by Lemma~\ref{lem:P4>=0}; the trace of the product of two positive operators is nonnegative. Thus, $G_{\omega}$ is a density operator.

For $\omega_{s}=\phi+s\Delta$, the derivative is $\frac{d}{ds}f(\omega_{s})=4\tr(G_{\omega_{s}}\Delta)$. Since $G_{\omega_{s}}$ is an effect and $\tr\Delta=0$, its magnitude is at most $4D(\phi+\Delta,\phi)$. Integration proves Eq.~\eqref{eq:pure_continuity}, and also the same Lipschitz bound for $f$ on all density operators.

To bound the roof, retain the linear term before averaging. With $\Delta_{j}=\psi_{j}-\phi$, Taylor's theorem and Pauli orthogonality give
\begin{align*}
 \frac{d^{2}}{ds^{2}}f(\phi+s\Delta_{j})
 &=\frac{12}{d}\sum_{P}[\tr(P\omega_{s})]^{2}[\tr(P\Delta_{j})]^{2}\\*
 &\leq12\tr(\Delta_{j}^{2}),
\end{align*}
where now $\omega_{s}=\phi+s\Delta_{j}$. Hence, using purity of $\psi_{j},\phi$,
\begin{align*}
 P_{2}(\psi_{j})&\leq P_{2}(\phi)+4\tr(G_{\phi}\Delta_{j})
                   +6\tr(\Delta_{j}^{2})\\
 &=P_{2}(\phi)+4\tr(G_{\phi}\Delta_{j})
                   +12[1-\tr(\psi_{j}\phi)].
\end{align*}
For every decomposition $\rho=\sum_{j}p_{j}\psi_{j}$, averaging yields
\begin{align*}
 \sum_{j}p_{j}P_{2}(\psi_{j})
 &\leq P_{2}(\phi)+4\tr[G_{\phi}(\rho-\phi)]\\*
 &\quad+12[1-\tr(\rho\phi)]\\
 &\leq P_{2}(\phi)+16D(\rho,\phi).
\end{align*}
Both $G_{\phi}$ and $\phi$ are effects, so each corresponding deviation is bounded by the trace distance. Take the supremum and use $\widehat{P}_{2}\leq1$. Every constant is independent of $N$.
\end{proof}

\begin{theorem}[Approximate adaptive synthesis bound]\label{thm:approximate}
\leavevmode\par\nobreak\noindent
For every $n$-qubit unitary $U$ and $0\leq\epsilon\leq1$,
\begin{equation}\label{eq:approx_amortized}
\begin{split}
 &t^{\rm stab}_{\epsilon}(U)\\
 &\quad\geq\left\lceil\frac{1}{c_{T}}\sup_{m,\ket{\psi}}
 \log_{2}\!\frac{P_{2}(\psi)}{\min\{1,P_{2}((U\ox I)\psi)+16\epsilon\}}\right\rceil,
\end{split}
\end{equation}
where $m\in\NN_{0}$ and $\ket{\psi}\in\cH_{n+m}$ is normalized. In particular, writing $p_{U}=P_{2}(\Phi_{U})$,
\begin{equation}\label{eq:approx_choi}
 t^{\rm stab}_{\epsilon}(U)\geq
 \left\lceil\frac{-\log_{2}\min\{1,p_{U}+16\epsilon\}}{c_{T}}\right\rceil.
\end{equation}
At $\epsilon=0$, Eq.~\eqref{eq:approx_amortized} recovers the full amortized bound, now allowing stabilizer ancillary states, measurements, and feedforward.
\end{theorem}
\begin{proof}
Fix a protocol with budget $t$ and an arbitrary pure $\psi_{AR}$, with $|A|=n$ and $|R|=m$. Supply $t$ copies of $\ket{T}$ at the start; unused copies may be discarded. Free ancillary preparations and all adaptive operations are included in a deterministic stabilizer map $\Lambda$ on $A$ and the resource registers, acting as the identity on $R$. Pure-input multiplicativity and the roof monotonicity proved above give
\[
 P_{2}(\psi)(3/4)^{t}
 =P_{2}(\psi\ox T^{\ox t})
 \leq\widehat{P}_{2}((\cE\ox\id_{R})(\psi)).
\]
The channel error implies
\[
 D((\cE\ox\id_{R})(\psi),(\cU\ox\id_{R})(\psi))\leq\epsilon
\]
for every finite $R$, including $m>n$. To see why the usual input-sized reference in the diamond norm suffices, use a Schmidt decomposition of $\psi_{AR}$: its reference support has dimension at most $2^{n}$, and the embedding into $R$ is an isometry preserving the output trace norm.

Apply Lemma~\ref{lem:continuity} to the joint output on $AR$, whose ideal target is pure. Taking logarithms gives the claimed bound for each $m,\psi$ with the same error constant. Taking suprema over the inputs and reference sizes and minimizing over protocols proves the theorem. Choosing $\psi=\Phi_{n}$, for which $P_{2}=1$, gives Eq.~\eqref{eq:approx_choi}.
\end{proof}

For comparison, define the unitary circuit approximation cost
$t^{\rm unit}_{\epsilon}(U)=\min_{V:\,D_{\diamond}(\cV,\cU)\leq\epsilon}t(V)$, where $D_{\diamond}$ is the half-diamond distance in Eq.~\eqref{eq:diamond_convention} and an empty feasible set has cost $+\infty$. State injection gives $t^{\rm stab}_{\epsilon}(U)\leq t^{\rm unit}_{\epsilon}(U)$, so any lower bound in the adaptive model also bounds this unitary-circuit $T$-count. The implemented Choi state of a unitary circuit is pure, so Eq.~\eqref{eq:pure_continuity} gives the sharper correction
\begin{equation}\label{eq:approx_unitary}
 t^{\rm unit}_{\epsilon}(U)\geq
 \left\lceil\frac{-\log_{2}\min\{1,p_{U}+4\epsilon\}}{c_{T}}\right\rceil.
\end{equation}
Both corrections are linear in $\epsilon$; the factor $16$ allows a mixed output, while $4$ suffices for a pure unitary output.

\subsection{What can be measured efficiently?}\label{sec:measurement}
The unrestricted optimization in $M_{\alpha}^{\cA}(U)$ is generally difficult. The Choi witness fixes a maximally entangled input and depends on the normalized Pauli moment $P_{\alpha}=2^{(1-\alpha)M_{\alpha}}$, which is the directly estimated quantity.

For a fixed odd integer $\alpha>1$, Bell measurements on ordinary copies give a bounded estimator of $P_{\alpha}$. The cited algorithm for even integer $\alpha$, including $\alpha=2$, additionally needs copies of the complex-conjugate state~\cite{Haug2024}. For Choi states the required conjugate state is $\ket{\Phi_{U}}^{*}=\ket{\Phi_{U^{*}}}$. Access to the inverse unitary suffices, since
\[
 \ket{\Phi_{U}}^{*}=(U^{*}\ox I)\ket{\Phi_{n}}
 =(I\ox U^{\dagger})\ket{\Phi_{n}}.
\]

For either bounded estimator, let $\eta$ be additive moment accuracy and $\delta$ the failure probability. For outcomes in $[-1,1]$, $L\geq2\eta^{-2}\ln(2/\delta)$ independent repetitions suffice by Hoeffding's inequality. Each repetition uses $O(\alpha)$ copies and $O(\alpha N)$ classical work for an $N$-qubit state~\cite{Haug2024}. A Choi state has $N=2n$ qubits; a Bell measurement on two such copies uses $4n$ qubits, with the pairs processed sequentially. These costs exclude implementing the state preparation and, when required, its conjugate. To estimate $M_{2}$ to additive accuracy $\Delta$, a sufficient moment accuracy is $\eta\leq p_{U}(1-2^{-\Delta})$. Thus the bounded-estimator guarantee scales as $O(p_{U}^{-2}\Delta^{-2}\log(1/\delta))$ for small $\Delta$, which can be exponential in the number of qubits.

A useful alternative needs only ordinary copies.\linebreak[1]
Cauchy--Schwarz with the distribution $\Xi_{P}$ gives
\begin{equation}\label{eq:odd_moment_relation}
\begin{aligned}
 P_{2}(\psi)^{2}&=\big(\mathbb{E}_{P\sim\Xi}|\bra{\psi} P\ket{\psi}|^{2}\big)^{2}\\
 &\leq\mathbb{E}_{P\sim\Xi}|\bra{\psi} P\ket{\psi}|^{4}=P_{3}(\psi).
\end{aligned}
\end{equation}
If a measurement returns $\widehat{p}_{3}$ with $|\widehat{p}_{3}-P_{3}(\Phi_{U})|\leq\eta$, put $p_{3}^{+}=\min\{1,\max\{0,\widehat{p}_{3}+\eta\}\}$. With the same confidence $1-\delta$,
\begin{equation}\label{eq:odd_certificate}
 t^{\rm stab}_{\epsilon}(U)\geq
 \left\lceil\frac{-\log_{2}\min\{1,\sqrt{p_{3}^{+}}+16\epsilon\}}{c_{T}}\right\rceil.
\end{equation}
This generally weaker witness separates sampling error $\eta$ from synthesis error $\epsilon$, and can be evaluated using ordinary Choi copies. If conjugate access permits estimation of $P_{2}$, replace $\sqrt{p_{3}^{+}}$ by the corresponding upper confidence bound $p_{2}^{+}$. These measurement statements concern copies of the pure target Choi state. For a noisy mixed state, the measured Pauli polynomial and the roof $\widehat{P}_{2}$ are distinct quantities.

\subsection{Examples with an explicit error budget}\label{sec:examples}
We apply Eq.~\eqref{eq:approx_choi} to the standard quantum Fourier transform,
\[
 \mathrm{QFT}_{n}\ket{j}=2^{-n/2}\sum_{k=0}^{2^{n}-1}e^{2\pi ijk/2^{n}}\ket{k},
 \qquad 0\leq j<2^{n}.
\]
Table~\ref{tab:qft} reports the resulting bounds for deterministic adaptive protocols. The $\epsilon=0$ column evaluates the formal zero-error witness; the positive-error columns bound finite-accuracy resource costs. For example, the QFT$_{4}$ witness is eight at $\epsilon=10^{-6}$ and seven at $\epsilon=10^{-4}$, using the linear correction $16\epsilon$. The Fourier transform is a basic subroutine in quantum algorithms~\cite{Shor_1997,Harrow2009}.

\begin{table}[!tbp]
\centering
\caption{Choi-based lower bounds on $t^{\rm stab}_{\epsilon}(\mathrm{QFT}_{n})$. Errors are half-diamond distances.}
\label{tab:qft}
\begin{ruledtabular}
\begin{tabular}{cccccc}
$n$ & $P_{2}(\Phi_{U})$ & $\epsilon=0$ & $10^{-8}$ & $10^{-6}$ & $10^{-4}$\\
\hline
3 & 0.226562500 & 6 & 6 & 6 & 6\\
4 & 0.132502806 & 8 & 8 & 8 & 7\\
5 & 0.081150899 & 9 & 9 & 9 & 9\\
6 & 0.050477707 & 11 & 11 & 11 & 11\\
\end{tabular}
\end{ruledtabular}
\end{table}

For controlled rotations we use the same convention $R_{z}(\theta)=e^{-i\theta Z/2}$ as in Fig.~\ref{fig:MARz}. Define
\[
 CCR_{z}(\theta)=\operatorname{diag}(1,1,1,1,1,1,e^{-i\theta/2},e^{i\theta/2}).
\]
Writing $c=\cos\theta$, its Choi fourth moment is
\begin{equation}\label{eq:ccrz_moment}
 P_{2}(\Phi_{CCR_{z}(\theta)})=\frac{151+84c+20c^{2}+c^{4}}{256}.
\end{equation}
At $\theta=\pi$ this equals $11/32$, yielding a lower bound of four also at $\epsilon=10^{-6}$. For other parameters, the bound is obtained directly from Eqs.~\eqref{eq:ccrz_moment} and~\eqref{eq:approx_choi}.

Consider also
\[
 U_{s}=\ketbra{0}{0}\ox I_{4}+\ketbra{1}{1}\ox\mathrm{SWAP}^{s},
 \qquad 0\leq s\leq1,
\]
where $\mathrm{SWAP}^{s}$ is the identity on the symmetric subspace and multiplies the antisymmetric state by $e^{i\pi s}$. This includes the Fredkin gate at $s=1$; fractional SWAP operations arise from exchange interactions~\cite{Fan_2005}. For $c=\cos(\pi s)$,
\begin{equation}\label{eq:cswap_moment}
 P_{2}(\Phi_{U_{s}})=\frac{109+84c+56c^{2}+7c^{4}}{256}.
\end{equation}
For example, at $s=3/4$ the moment is $(555-168\sqrt{2})/1024$, and Eq.~\eqref{eq:approx_choi} gives $t^{\rm stab}_{10^{-6}}(U_{3/4})\geq5$. At $s=1$ it gives four. At $s=1/2$ the Choi witness gives three. More generally, the zero-error five-gate certificate holds precisely when
$c^{4}+8c^{2}+12c+4<0$; at nonzero error the condition is instead $P_{2}(\Phi_{U_{s}})+16\epsilon<(3/4)^{4}$.
The formulas are obtained from the Pauli transfer expression
\begin{equation}\label{eq:choi_ptm}
 P_{2}(\Phi_{U})=\frac{1}{d^{2}}\sum_{P,Q\in\cP_{n}}
 \left[\frac{\tr(PUQU^{\dagger})}{d}\right]^{4},\qquad d=2^{n}.
\end{equation}
The transpose on the reference Pauli in the Choi identity only changes a sign and hence does not affect the fourth power. Appendix~\ref{app:numerics} specifies an independent exact check of these polynomial expressions. Controlled SWAP is relevant to state-purification protocols~\cite{Cirac1999,Yao2025,He2025}. The unitary stabilizer nullity~\cite{Jiang2023} provides a comparator for zero-error synthesis.

Finally, consider the four-site periodic Heisenberg chain of the original illustration,
\begin{equation}\label{eq:heisenberg}
\begin{split}
 H={}&\sum_{k=1}^{N}(X_{k}X_{k+1}+Y_{k}Y_{k+1}+\Delta Z_{k}Z_{k+1})
       +\sum_{k=1}^{N}h_{k}Z_{k},\\
 &N=4,\quad\Delta=0.2,\quad h_{k}\sim\operatorname{Unif}[-W,W],
\end{split}
\end{equation}
where $X_{N+1}=X_{1}$, and similarly for $Y,Z$. The exchange coefficient is one, and the fields are independent. We retain $W=0.5,1,2,3$ and compare the zero-error Choi witness with the unitary stabilizer nullity~\cite{Jiang2023}. Let $s(U)$ count the ordered pairs $P,Q\in\cP_{N}$ satisfying $UQU^{\dagger}=\pm P$. Then
\[
 \nu(U)=2N-\log_{2}s(U).
\]
Every term of $H$ commutes with $\Pi=Z_{1}Z_{2}Z_{3}Z_{4}$. Hence $U(t)\Pi U(t)^{\dagger}=\Pi$, and the identity and $\Pi$ already give $s(U(t))\geq2$. Consequently, $\nu(U(t))\leq7$ for every field realization and time. The line at seven in Fig.~\ref{fig:t-count_lower_bounds} is this analytic \emph{ceiling on the nullity witness}. Whenever a sample's zero-error SRE witness exceeds seven, it is strictly stronger than its nullity witness. For each of the four $W$ values, the zero-error sample mean exceeds seven at all sampled times $t\geq10$.

\begin{figure}[!tbp]
\centering
\includegraphics[width=\linewidth]{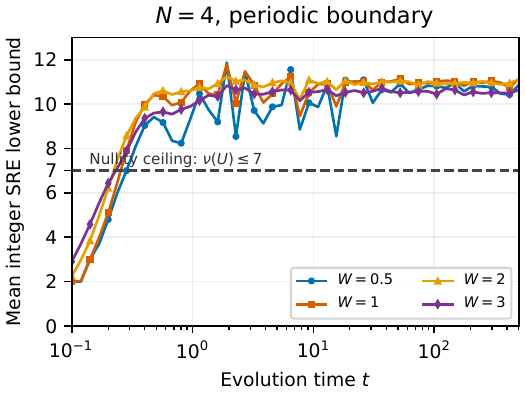}
\caption{Four-site periodic Heisenberg dynamics in Eq.~\eqref{eq:heisenberg}. Curves show the sample mean of the per-realization integer bound in Eq.~\eqref{eq:approx_choi} at zero error across 200 independent field vectors per $W$. The dashed line is the analytic ceiling of seven on the nullity witness. The calculation uses fresh realizations of the original model, with NumPy PCG64 seed 20261021, successive draws for $W=0.5,1,2,3$, and 50 logarithmic times from $0.1$ to $500$.}
\label{fig:t-count_lower_bounds}
\end{figure}

At a prescribed half-diamond tolerance $\epsilon=10^{-4}$, the same target moments give approximate-cost lower bounds through Eq.~\eqref{eq:approx_choi}. Of 40,000 paired sample/time entries, 36,426 integer bounds remain unchanged, and the others decrease by one; the corresponding arrays are supplied. This quantifies the sensitivity of the lower bound to the prescribed tolerance. The open three-site calculation is retained in Appendix~\ref{app:numerics}, Fig.~\ref{fig:heisenberg_three_site}.

\section{Concluding remarks}\label{sec:conclusion}

Amortized SRE captures the largest net nonstabilizerness increase of a unitary. The exact $T$- and $CCZ$-gate results calibrate this quantity, while the explicit rotation witness shows that initial magic can increase its value. The corrected extent argument confirms that this dependence is absent for the two decomposition-based logarithmic measures considered here.

Our error-dependent resource bound applies to deterministic stabilizer protocols with a worst-case $T$-state budget, including ancillary states, measurement, and feedback. Its Choi specialization gives finite-accuracy witnesses even when exact synthesis is unavailable. Ordinary-copy odd-moment estimation provides one practical route. The unrestricted optimization, estimation under other access models, and sharper error dependence are directions for further work. A further question is whether $M_{\alpha}^{\cA}(U^{\dagger})=M_{\alpha}^{\cA}(U)$ always holds; a failure would distinguish magic generation and destruction, as happens for entangling and disentangling capacities~\cite{Linden2009}.

\begin{acknowledgments}
The authors thank Sixuan Wu, Xia Liu, and Chenghong Zhu for their comments. This work was partially supported by the National Key R\&D Program of China (Grant No.~2024YFB4504004), the National Natural Science Foundation of China (Grant No.~12447107, 92576114), the Guangdong Provincial Quantum Science Strategic Initiative (Grant No.~GDZX2403008, GDZX2503001), the CCF-Tencent Rhino-Bird Open Research Fund, and the Guangdong Provincial Key Lab of Integrated Communication, Sensing and Computation for Ubiquitous Internet of Things (Grant No.~2023B1212010007). 
\end{acknowledgments}

\section*{Data Availability}
The MATLAB code, Python verification scripts, exact symbolic calculations, interval certificates, and raw numerical data are available at~\cite{coderepo}.
\appendix
\section{Exact amortized 2-SRE of CCZ}\label{app:amor_magic_CCZ}

We prove the fourth-moment ratio bound by decomposing the adjoint action of $CCZ$ into invariant Pauli subspaces.
Let $\ket{\psi}\in\cH_{3+m}$, write
$p_{PQ}=\bra{\psi}(P\ox Q)\ket{\psi}$, and define
\[
 R=\sum_{P\in\cP_{3},Q\in\cP_{m}}p_{PQ}^{4},
 \qquad A=\sum_{P\in\{I,Z\}^{\ox3},Q\in\cP_{m}}p_{PQ}^{4}.
\]
Here $m\in\NN_{0}$ is arbitrary.

First, $A\geq R/8$. On four copies of one qubit let $D=I^{\ox4}+Z^{\ox4}$ and use $Q_{\pm}$ from Lemma~\ref{lem:P4>=0}. Since $2D-Q_{+}=Q_{-}\succeq0$, we have $0\preceq Q_{+}\preceq2D$. Tensoring preserves this order: for positive $B\succeq C$, the difference $B^{\ox3}-C^{\ox3}$ is the sum of the three positive terms $C^{\ox j}\ox(B-C)\ox B^{\ox(2-j)}$, $j=0,1,2$. Consequently,
\[
 Q_{+}^{\ox3}\ox Q_{+}^{\ox m}
 \preceq8D^{\ox3}\ox Q_{+}^{\ox m}.
\]
Take the expectation in $\ket{\psi}^{\ox4}$, after the fixed permutation grouping the four replicas of each qubit. The two expectations are $R$ and $8A$, respectively. For $m=0$, the reference factor is the scalar identity.

Next, use binary labels $x,z\in\FF_{2}^{3}$ for the Hermitian Paulis $P_{x,z}=i^{\sum_{j}x_{j}z_{j}}X^{x} Z^{z}$. For each nonzero $x$, the eight values of $z$ split into the two cosets of
$x^{\perp}=\{k:k\cdot x=0\pmod2\}$. Conjugation by $CCZ$ preserves each of these four-dimensional Pauli spans. To verify this explicitly, put $a=(1,1,1)$ and write $CCZ\ket{y}=d(y)\ket{y}$, where $d(a)=-1$ and $d(y)=1$ otherwise. For $x\ne0$,
\[
 d(y)d(y+x)=1-2\delta_{y,a}-2\delta_{y,a+x}.
\]
For $g_{x}(y)=d(y)d(y+x)$, the normalized binary Fourier coefficient is
\begin{equation}\label{eq:ccz_fourier}
\begin{aligned}
 \widehat{g}_{x}(k)&=\tfrac{1}{8}\sum_{y}(-1)^{k\cdot y}g_{x}(y)\\
 &=\delta_{k,0}-\tfrac{1}{4}(-1)^{k\cdot a}[1+(-1)^{k\cdot x}].
\end{aligned}
\end{equation}
This vanishes outside $x^{\perp}$ and equals $\pm1/2$ on $x^{\perp}$. More explicitly, with $s(x,z)=\sum_{j}x_{j}z_{j}$ evaluated as an integer,
\[
\begin{aligned}
 CCZ\,P_{x,z}\,CCZ
 &=\sum_{k\in x^{\perp}}\widehat{g}_{x}(k)\\
 &\quad\times i^{s(x,z)-s(x,z\oplus k)}P_{x,z\oplus k}.
\end{aligned}
\]
The exponent is even because $x\cdot k=0$ modulo two, so the extra phase is a real sign. Each of the two cosets of $x^{\perp}$ is therefore a four-dimensional invariant span. The adjoint action has eight fixed diagonal Paulis ($x=0$) and fourteen real $4\times4$ blocks. These blocks are orthogonal because the Pauli basis is Hilbert--Schmidt orthogonal and unitary conjugation preserves that inner product.

For any such block $O$ and real coefficient vector $v$,
\[
 \|Ov\|_{4}^{4}\geq\tfrac{1}{4}\|Ov\|_{2}^{4}
 =\tfrac{1}{4}\|v\|_{2}^{4}\geq\tfrac{1}{4}\|v\|_{4}^{4}.
\]
Applying this to each block separately for each reference Pauli $Q$ gives
\begin{equation}\label{eq:ccz_ratio}
 R_{2}((CCZ\ox I)\psi)\geq A+\tfrac{1}{4}(R-A)
 \geq\tfrac{11}{32}R.
\end{equation}
This proves the upper bound uniformly over all inputs and references. For attainment, $\ket{+}^{\ox3}$ has $R=8$. After $CCZ$, its Pauli expectations consist of one identity coefficient of magnitude one, 28 coefficients of magnitude $1/2$ (four for each $x\ne0$), and 35 zeros. Thus $P_{2}=(1+28/16)/8=11/32$. The bound is attained and
\[
 M_{2}^{\cA}(CCZ)=\log_{2}(32/11)=5-\log_{2} 11.
\]
The block structure and the one-qubit positive-operator inequalities were additionally checked with exact rational arithmetic in the accompanying verification code.

\section{Reproducible calculations and numerical checks}\label{app:numerics}

For a pure $N$-qubit state with amplitudes $\psi_{b}$ in binary order, the Pauli expectation for labels $x,z\in\FF_{2}^{N}$ is, up to a phase of modulus one, the binary Walsh--Hadamard transform in $b$ of
$\psi_{b}\psi_{b\oplus x}^{*}$. Thus, all even absolute Pauli moments can be calculated in $O(N4^{N})$ arithmetic operations. Our implementation processes $x$ in batches and uses this identity independently of the original MATLAB implementation~\cite{coderepo}. Direct Pauli enumeration for small states agrees to within $3\times10^{-16}$ in the verification run.

For QFT Choi states, $N=2n$, $d=2^{n}$, and $\psi_{a,b}=e^{2\pi iab/d}/d$. The products needed by the Walsh transform are exactly $d^{-2}$ times $d$th roots of unity. We enclose these roots with 60-decimal-digit interval arithmetic, propagate outward-rounded binary64 intervals through the transform, and bound positive summations with a conservative floating-point accumulation factor. This gives upper enclosures for $P_{2}$, which are compared with the exact rational thresholds $(3/4)^{k}$ after adding an outward-rounded continuity correction. These comparisons certify the integer entries of Table~\ref{tab:qft} independently of the displayed decimal rounding.

For Eqs.~\eqref{eq:ccrz_moment} and~\eqref{eq:cswap_moment}, symbolic evaluation of Eq.~\eqref{eq:choi_ptm} gives the Laurent polynomials
\begin{align*}
 P_{2}(\Phi_{CCR_{z}})={}&\frac{1291}{2048}+\frac{21}{128}(z^{2}+z^{-2})\\
 &+\frac{21}{1024}(z^{4}+z^{-4})+\frac{1}{4096}(z^{8}+z^{-8}),\\
 &z=e^{i\theta/2},\\
 P_{2}(\Phi_{U_{s}})={}&\frac{1117}{2048}+\frac{21}{128}(z+z^{-1})\\
 &+\frac{63}{1024}(z^{2}+z^{-2})+\frac{7}{4096}(z^{4}+z^{-4}),\\
 &z=e^{i\pi s}.
\end{align*}
The accompanying script performs the traces with exact symbolic arithmetic and also checks a nonspecial parameter against the independent Walsh implementation. In particular, the statement $t^{\rm stab}_{10^{-6}}(U_{3/4})\geq5$ follows by substituting $(555-168\sqrt{2})/1024$ into Eq.~\eqref{eq:approx_choi} and comparing with $81/256$.

For Fig.~\ref{fig:t-count_lower_bounds}, we diagonalize each stored $16\times16$ periodic Hamiltonian, construct $e^{-iHt}$, and evaluate the Choi moment. An independent implementation constructs the Hamiltonian from spin-basis exchange rules and checks every evolution using a matrix exponential and direct Pauli-transfer traces. The analytic nullity ceiling follows from the conserved parity.

The dynamics calculations use binary64 arithmetic and a downward guard of $10^{-11}$ on the real-valued cost witness before taking its integer ceiling. The dynamics were cross-checked using independent binary64 implementations; Table~\ref{tab:qft} has the interval guarantees described above. The supplied random-state tests serve as diagnostics for the analytic $T$ and $CCZ$ proofs.

\subsection{Supplementary open three-site example}
We retain an additional check with an open chain,
\[
 H_{3}=\sum_{k=1}^{2}(X_{k}X_{k+1}+Y_{k}Y_{k+1}+0.2Z_{k}Z_{k+1})
       +\sum_{k=1}^{3}h_{k}Z_{k}.
\]
Here $h_{k}$ are independent and uniform on $[-W,W]$, with $W=1,3$. Fig.~\ref{fig:heisenberg_three_site} uses 200 field vectors per $W$, sampled with NumPy PCG64 seed 20261018 and successive draws, at $t=0,0.05,\ldots,3$. The same realizations are used for all errors and times at fixed $W$. This independent finite-size example illustrates the error dependence of Eq.~\eqref{eq:approx_choi}. Its $8\times8$ Hamiltonians and all 24,400 evolution points were also checked by spin-basis construction, matrix exponentiation, and direct Pauli-transfer traces.

\begin{figure}[!tbp]
\centering
\includegraphics[width=\linewidth]{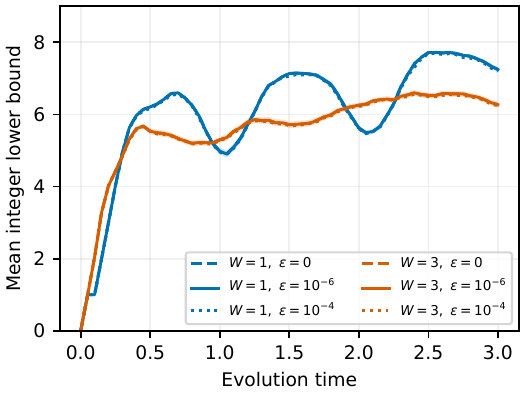}
\caption{Supplementary open three-site calculation. Blue: $W=1$; orange: $W=3$. Dashed, solid, and dotted curves use $\epsilon=0,10^{-6},10^{-4}$, respectively, in Eq.~\eqref{eq:approx_choi}. Each curve is the mean of 200 per-realization integer lower bounds; shading on solid curves is one standard error across disorder.}
\label{fig:heisenberg_three_site}
\end{figure}

\bibliographystyle{apsrev4-2}
\bibliography{ref}
\end{document}